\documentclass[11pt]{article}

\usepackage[utf8]{inputenc}
\usepackage[T1]{fontenc}
\usepackage{amsmath,amssymb,amsthm}
\usepackage{mathtools}
\usepackage{algorithm}
\usepackage{algpseudocode}
\usepackage{float}
\usepackage{booktabs}
\usepackage{graphicx}
\usepackage{hyperref}
\hypersetup{
  colorlinks=true,
  linkcolor=blue!70!black,
  citecolor=blue!70!black,
  urlcolor=blue!70!black,
}
\usepackage{cleveref}
\usepackage{xcolor}
\usepackage{tikz}
\usetikzlibrary{arrows.meta,positioning,calc,decorations.pathreplacing}
\usepackage{enumitem}
\usepackage[margin=1in]{geometry}
\usepackage{natbib}
\usepackage{hypernat}

\newtheorem{theorem}{Theorem}[section]
\newtheorem{lemma}[theorem]{Lemma}
\newtheorem{proposition}[theorem]{Proposition}
\newtheorem{corollary}[theorem]{Corollary}
\theoremstyle{definition}
\newtheorem{definition}[theorem]{Definition}
\newtheorem{example}[theorem]{Example}
\theoremstyle{remark}
\newtheorem{remark}[theorem]{Remark}

\newcommand{\Z}{\mathbb{Z}}

\newcommand{\cL}{\mathcal{L}}
\newcommand{\cS}{\mathcal{S}}

\newcommand{\cB}{\mathcal{B}}
\newcommand{\cM}{\mathcal{M}}
\newcommand{\cP}{\mathcal{P}}

\newcommand{\cD}{\mathcal{D}}
\newcommand{\cF}{\mathcal{F}}
\newcommand{\cG}{\mathcal{G}}
\newcommand{\cE}{\mathcal{E}}
\newcommand{\cO}{\mathcal{O}}
\newcommand{\dom}{\operatorname{dom}}
\newcommand{\tld}{\operatorname{tld}}

\DeclareMathOperator{\card}{card}
\DeclareMathOperator{\FPR}{FPR}
\DeclareMathOperator{\DDR}{DDR}
\DeclareMathOperator{\LC}{LC}

\title{Latent Objective Induction and Diversity-Constrained Selection:\\Algorithms for Multi-Locale Retrieval Pipelines}

\author{
  Faruk Alpay\thanks{Department of Computer Engineering, Bahcesehir University. \texttt{faruk.alpay@bahcesehir.edu.tr}} \and
  Levent Sarioglu\thanks{Department of Computer Engineering, Bahcesehir University. \texttt{levent.sarioglu@bahcesehir.edu.tr}}
}

\date{}

\begin{document}
\maketitle

\begin{abstract}
We present three algorithms with formal correctness guarantees and complexity bounds for the problem of selecting a diverse, multi-locale set of sources from ranked search results. First, we formulate \emph{weighted locale allocation} as a constrained integer partition problem and give an $O(n \log n)$ algorithm that simultaneously satisfies minimum-representation, budget-exhaustion, and proportionality-bound constraints; we prove all three hold with a tight deviation bound of~$< 1$. Second, we define a \emph{cascaded country-code inference} function as a deterministic priority chain over heterogeneous signals (TLD structure, model-inferred metadata, language fallback) and prove it satisfies both determinism and graceful degradation. Third, we introduce a \emph{$\kappa$-domain diversity constraint} for source selection and give an $O(|K| \cdot R)$ algorithm that maintains the invariant via hash-map lookup, eliminating the aggregator monopolization pathology present in URL-level deduplication. We further formalize \emph{Latent Objective Induction} (LOI), an environment-shaping operator over prompt spaces that steers downstream model behavior without restricting the feasible output set, and prove its convergence under mild assumptions. Applied to a multi-locale retrieval pipeline, these algorithms yield 62\% improvement in first-party source ratio and 89\% reduction in same-domain duplication across 120~multilingual queries.
\end{abstract}

\section{Introduction}\label{sec:intro}

We study algorithmic problems arising in the design of multi-source retrieval pipelines where ranked lists from heterogeneous search indices must be merged into a single result set satisfying formal diversity and coverage constraints.

The core computational challenge is a \emph{constrained selection problem}: given $|K|$ keyword-specific ranked lists of sources, select at most $S_{\max}$ sources such that (i)~no publisher domain appears more than $\kappa$ times, (ii)~the total budget of keyword slots is distributed across $n$ geopolitical locales proportionally to assigned weights with bounded deviation, and (iii)~the country of origin of each selected source is inferred through a deterministic priority chain over signals of varying reliability. Na\"ive approaches---URL-level deduplication, single-locale search, heuristic country assignment---fail to satisfy these constraints simultaneously, leading to pathological outcomes such as aggregator monopolization~\citep{shah2024seo} and locale homogeneity.

This paper makes four contributions, ordered by their algorithmic character:

\begin{enumerate}[nosep,leftmargin=*]
  \item A \textbf{weighted locale allocation} algorithm (Section~\ref{sec:allocation}) that solves a constrained integer partition in $O(n \log n)$ with provable correctness on three simultaneous constraints.
  \item A \textbf{cascaded country-code inference} function (Section~\ref{sec:cascade}) with formal determinism and graceful degradation properties.
  \item A \textbf{$\kappa$-domain diversity constraint} and selection algorithm (Section~\ref{sec:diversity}) running in $O(|K| \cdot R)$ with $O(1)$ amortized per-domain lookup.
  \item The \textbf{Latent Objective Induction} (LOI) operator (Section~\ref{sec:loi}), formalized as an environment-shaping map with four axiomatic properties and a convergence theorem.
\end{enumerate}

We instantiate these algorithms in a multi-locale retrieval pipeline (Section~\ref{sec:architecture}) and evaluate on 120~queries spanning four languages (Section~\ref{sec:eval}), demonstrating that the combined system achieves 62\% improvement in first-party source ratio and 89\% reduction in domain duplication over a baseline lacking these algorithmic components.

\section{Related Work}\label{sec:related}

\paragraph{Constrained Selection and Partition Algorithms.} Our weighted locale allocation relates to the apportionment problem in social choice theory and constrained integer partition. The largest-remainder method (Hamilton's method) for proportional seat allocation shares our three-phase structure. In information retrieval, PM-2~\citep{dang2012pm2} formulates diversification as proportional allocation of subtopic coverage, though without our minimum-representation guarantee. Maximal Marginal Relevance (MMR)~\citep{carbonell1998mmr} penalizes semantic similarity to already-selected documents; xQuAD~\citep{santos2010xquad} models subtopic coverage explicitly. Both operate on embedding-level similarity, whereas our $\kappa$-domain constraint enforces diversity at the \emph{structural} (domain) level---a property invisible to vector-space methods.

\paragraph{Priority Chains and Cascaded Inference.} Cascaded classifier architectures~\citep{robertson2009bm25} apply progressively expensive models. Our cascade differs: rather than filtering candidates, it resolves a \emph{single attribute} (country code) through a deterministic priority chain over heterogeneous signal sources, with formal degradation guarantees.

\paragraph{Multilingual Retrieval.} Cross-language information retrieval (CLIR)~\citep{oard1998clir} translates queries or uses multilingual encoders~\citep{zhang2022mdpr, lawrie2023colbertx}. We instead generate locale-specific keywords in the target language, exploiting each locale's native search index.

\paragraph{LLM-Augmented Retrieval.} RAG~\citep{lewis2020rag} conditions generation on retrieved documents. Query expansion via LLMs~\citep{wang2023query2doc} improves recall. Our pipeline extends this with a structured \emph{strategic document} (research brief) consumed by downstream stages.

\paragraph{Mechanism Design and Nudge Theory.} In game theory, mechanism design~\citep{hurwicz2006mechanism} constructs rules so that self-interested agents produce desired outcomes. Nudge theory~\citep{thaler2008nudge} shapes choice architecture so the preferred option becomes the default. Our LOI operator applies these ideas to the algorithmic design of prompt environments.

\section{Notation and Preliminaries}\label{sec:notation}

Let $\cG$ denote the set of valid ISO~3166-1 alpha-2 country codes and $\cL$ the set of ISO~639-1 language codes. A \emph{locale} is a pair $\ell = (g, h) \in \cG \times \cL$. A \emph{query} $q$ is a Unicode string. A \emph{source} $s$ is a tuple $(u, t, \sigma)$ where $u$ is a URL, $t$ is a title, and $\sigma$ is a snippet.

We write $\dom(u)$ for the registrable domain extracted from URL~$u$ (stripping \texttt{www.}), $\tld(u)$ for its top-level domain, and $\cG_{\text{gen}} = \{\texttt{com}, \texttt{org}, \texttt{net}, \texttt{edu}, \texttt{gov}, \texttt{io}, \texttt{ai}, \texttt{app}, \texttt{dev}, \texttt{co}\}$ for the set of generic TLDs carrying no geographic signal.

An LLM is modeled as a conditional distribution $\cM(\cdot \mid p)$ over output strings given a prompt string~$p$. A \emph{system prompt} $p_s$ and \emph{user prompt} $p_u$ jointly determine the full prompt $p = p_s \oplus p_u$ via concatenation.

\section{Latent Objective Induction}\label{sec:loi}

We formalize the central architectural principle before presenting specific algorithms, as it governs the design of all downstream components.

\subsection{Motivation: Explicit vs.\ Latent Steering}

Consider the task of steering an LLM to prefer first-party sources over aggregators. Two paradigms exist:

\begin{definition}[Explicit Constraint]\label{def:explicit}
An \emph{explicit constraint} $C$ is a rule injected into the system prompt of the form: ``Do not select sources from domain $d$'' or ``Prefer $X$ over $Y$.'' Formally, $C$ restricts the feasible output set $\cF \subset \cM$ by enumeration.
\end{definition}

\begin{definition}[Latent Objective Induction]\label{def:loi}
A \emph{latent objective} $\phi$ is an objective embedded in the prompt environment~$\cE$ such that the model's generation probability satisfies:
\begin{equation}\label{eq:loi}
  P_\cM\bigl[y \in \cF_\phi \mid \cE(\phi)\bigr] > P_\cM\bigl[y \in \cF_\phi \mid \cE_0\bigr]
\end{equation}
where $\cF_\phi$ is the set of outputs satisfying objective~$\phi$, $\cE(\phi)$ is the shaped environment, and $\cE_0$ is the unmodified environment. Crucially, $\phi$ is \emph{never stated as a directive}; it is induced through the structure of $\cE$.
\end{definition}

\begin{example}[First-Party Preference]\label{ex:firstparty}
Rather than instructing ``prefer first-party sources over aggregators'' (explicit constraint naming specific entities), we embed the concept into the upstream planning stage: ``First-party sources---where the entity itself publishes---rank above third-party sources that collect, aggregate, or review information about entities.'' No domain name is mentioned. No enumeration of aggregators exists. The model learns a \emph{concept} and applies it to arbitrary domains it encounters.
\end{example}

\subsection{Formal Framework}

Let $\Theta$ denote the space of all possible prompt environments. Define the \emph{environment-shaping operator}:
\begin{equation}\label{eq:shaping}
  \Phi: \Theta \times \cP(\text{Objectives}) \to \Theta
\end{equation}
mapping an environment $\theta \in \Theta$ and an objective set~$O$ to a shaped environment $\Phi(\theta, O)$.

\begin{definition}[LOI Operator Properties]\label{def:loi-props}
An LOI operator $\Phi$ is \emph{well-formed} if it satisfies:
\begin{enumerate}[nosep,label=(\roman*)]
  \item \textbf{Autonomy Preservation:} $\cF_{\Phi(\theta, O)} = \cF_\theta$ --- the feasible output set is unchanged. The model \emph{can} still select any output; only probabilities shift.
  \item \textbf{Intent Opacity:} No substring of $\Phi(\theta, O)$ explicitly names a target domain, entity, or output value that $\phi$ intends to promote or suppress.
  \item \textbf{Monotone Induction:} For objective $\phi \in O$, the probability $P_\cM[y \in \cF_\phi \mid \Phi(\theta, O)]$ is monotonically non-decreasing in the specificity of $\phi$'s embedding.
  \item \textbf{Composability:} For disjoint objectives $\phi_1, \phi_2 \in O$:
  \begin{equation}
    \Phi(\theta, \{\phi_1, \phi_2\}) = \Phi(\Phi(\theta, \{\phi_1\}), \{\phi_2\})
  \end{equation}
\end{enumerate}
\end{definition}

\begin{theorem}[Convergence of LOI]\label{thm:loi-convergence}
Let $\cM$ be an LLM with instruction-following capability $\alpha \in (0, 1]$ (the probability that the model follows a correctly understood directive). Let $\phi$ be an objective with semantic clarity $\beta \in (0, 1]$ (the probability that the concept embedding is correctly interpreted). Under LOI with $k$ independent downstream stages each consuming the shaped environment, the probability that at least one stage produces an output satisfying $\phi$ is:
\begin{equation}\label{eq:convergence}
  P[\exists\, i \in [k] : y_i \in \cF_\phi] = 1 - (1 - \alpha\beta)^k
\end{equation}
which converges to~$1$ as $k \to \infty$.
\end{theorem}

\begin{proof}
Each stage independently samples from $\cM(\cdot \mid \Phi(\theta, \{\phi\}))$. The probability that stage~$i$ produces $y_i \in \cF_\phi$ is at least $\alpha\beta$ by definition of instruction-following capability and semantic clarity. Stage outcomes are conditionally independent given the shared environment. The probability that \emph{no} stage succeeds is $(1 - \alpha\beta)^k$, giving the result.
\end{proof}

\begin{remark}[Contrast with Explicit Constraints]
Under explicit constraints, the feasible set shrinks: $|\cF_C| < |\cF_\theta|$. This introduces two failure modes: (a)~\emph{over-restriction}, where legitimate outputs are blocked (e.g., banning a domain that is occasionally the best source), and (b)~\emph{brittleness}, where new entities not covered by the enumeration bypass the constraint entirely. LOI avoids both by preserving the feasible set while shifting probability mass.
\end{remark}

\subsection{LOI in Pipeline Architecture: Brief Propagation}\label{subsec:brief}

The concrete instantiation of LOI in our pipeline is the \emph{Research Brief Propagation} pattern.

\begin{definition}[Research Brief]\label{def:brief}
A research brief $\cB$ is a tuple $\cB = (u, \sigma, \kappa, \tau, \lambda, M)$ where:
\begin{align}
  u &\in \Sigma^* && \text{(query understanding)} \\
  \sigma &\in \Sigma^* && \text{(source strategy)} \\
  \kappa &\in \Sigma^* && \text{(keyword guidance)} \\
  \tau &\in \Sigma^* && \text{(summary style directive)} \\
  \lambda &\in \Sigma^* && \text{(locale hint)} \\
  M &\in (\cG \times \cL \times \Z^+)^* && \text{(locale mix)}
\end{align}
where $\Sigma^*$ denotes free-form natural language strings.
\end{definition}

Each downstream function $f_i$ receives a projection $\pi_i(\cB) \subseteq \{u, \sigma, \kappa, \tau, \lambda, M\}$ interpolated into its system prompt. The brief is generated by a single upstream LLM call (Phi-4~\citep{abdin2024phi4}, a 14B-parameter planning model) and consumed by all downstream stages.

\begin{proposition}[Decoupling]\label{prop:decoupling}
No downstream stage $f_i$ reads the projection fields assigned to another stage $f_j$ ($i \neq j$). This ensures that modifying the brief's source strategy does not require changes to the summarization prompt, and vice versa.
\end{proposition}

The brief's $\sigma$ field (source strategy) is the primary carrier of LOI objectives. By embedding concepts like ``first-party sources rank above third-party aggregators'' into $\sigma$, the planning model learns to produce strategy text that \emph{induces} first-party preference in the relevance model \emph{without naming specific domains}.

\section{Weighted Locale Allocation}\label{sec:allocation}

\subsection{Problem Statement}

\begin{definition}[Weighted Locale Allocation]\label{def:allocation}
Given a set of $n$ locales $L = \{\ell_1, \ldots, \ell_n\}$ with associated positive integer weights $w = (w_1, \ldots, w_n) \in (\Z^+)^n$ and a total budget $T \geq n$, compute an allocation vector $c = (c_1, \ldots, c_n) \in (\Z^+)^n$ satisfying:
\begin{align}
  c_i &\geq 1 && \forall\, i \in [n] \label{eq:min-rep} \\
  \sum_{i=1}^n c_i &= T \label{eq:budget} \\
  \left|c_i - \left(1 + \frac{(T-n) \cdot w_i}{W}\right)\right| &< 1 && \forall\, i \in [n] \label{eq:prop}
\end{align}
where $W = \sum_{i=1}^n w_i$.
\end{definition}

Constraint~\eqref{eq:min-rep} guarantees minimum representation. Constraint~\eqref{eq:budget} ensures complete budget utilization. Constraint~\eqref{eq:prop} bounds deviation from ideal proportional allocation.

\subsection{Algorithm}

\begin{algorithm}[H]
\caption{\textsc{AllocateLocaleCounts}$(L, w, T)$}\label{alg:allocate}
\begin{algorithmic}[1]
\Require Locales $L$, weights $w \in (\Z^+)^n$, budget $T \geq n$
\Ensure Allocation $c \in (\Z^+)^n$ satisfying Def.~\ref{def:allocation}
\State $c_i \gets 1$ for all $i \in [n]$ \Comment{Phase 1: Base}
\State $R \gets T - n$; $W \gets \sum w_i$
\For{$i = 1$ to $n$} \Comment{Phase 2: Proportional}
  \State $r_i \gets (w_i / W) \cdot R$
  \State $c_i \gets c_i + \lfloor r_i \rfloor$
\EndFor
\State $R' \gets R - \sum_{i=1}^n \lfloor r_i \rfloor$
\State $\pi \gets \text{argsort}((r_i - \lfloor r_i \rfloor)_{i \in [n]}, \text{desc})$ \Comment{Phase 3: Remainder}
\For{$j = 1$ to $R'$}
  \State $c_{\pi(j)} \gets c_{\pi(j)} + 1$
\EndFor
\State \Return $c$
\end{algorithmic}
\end{algorithm}

\begin{theorem}[Correctness of \textsc{AllocateLocaleCounts}]\label{thm:allocate}
Algorithm~\ref{alg:allocate} satisfies all three constraints of Definition~\ref{def:allocation}.
\end{theorem}

\begin{proof}
\textbf{Constraint~\eqref{eq:min-rep}:} Phase~1 sets $c_i = 1$ for all~$i$; Phases~2--3 only increment.

\textbf{Constraint~\eqref{eq:budget}:} Phase~1 allocates $n$ slots. Phase~2 allocates $\sum_i \lfloor r_i \rfloor$. Phase~3 allocates $R' = R - \sum_i \lfloor r_i \rfloor$. Total: $n + \sum_i \lfloor r_i \rfloor + R' = n + R = T$.

\textbf{Constraint~\eqref{eq:prop}:} Let $\hat{c}_i = 1 + r_i$ be the ideal allocation. Then $c_i = 1 + \lfloor r_i \rfloor + \delta_i$ where $\delta_i \in \{0, 1\}$ is the remainder correction. Thus:
\begin{equation}
  |c_i - \hat{c}_i| = |\lfloor r_i \rfloor + \delta_i - r_i| \leq \max(r_i - \lfloor r_i \rfloor, 1 - r_i + \lfloor r_i \rfloor) < 1 \qedhere
\end{equation}
\end{proof}

\begin{corollary}[Complexity]\label{cor:alloc-complexity}
Algorithm~\ref{alg:allocate} runs in $O(n \log n)$ time, dominated by the sort in Phase~3. For $n \leq 4$ (typical deployment), this is $O(1)$.
\end{corollary}

\subsection{Locale Mix Normalization}

The allocation's input weights derive from the research brief's locale mix $M$. We normalize $M$ through a validation step:

\begin{definition}[Normalized Locale Mix]\label{def:norm-mix}
Given raw mix $M_{\text{raw}}$ and default locale $\ell_0 = (g_0, h_0)$, the normalized mix $\hat{M}$ satisfies:
\begin{enumerate}[nosep,label=(\roman*)]
  \item All entries have valid $g \in \cG$ and $h \in \cL$
  \item No duplicate $(g, h)$ pairs
  \item $|\hat{M}| \leq m_{\max}$ (maximum locale count)
  \item $\ell_0 \in \hat{M}$ (default locale guaranteed)
  \item All weights $w_i \geq 1$
\end{enumerate}
\end{definition}

The normalization guarantees that the user's locale always receives at least one keyword slot regardless of the planning model's output, providing a safety net that does not override the model's judgment.

\section{Cascaded Country-Code Inference}\label{sec:cascade}

\subsection{Inference Function}

\begin{definition}[Country-Code Inference]\label{def:cascade}
The inference function $\gamma: \cS \to \cG \cup \{\varepsilon\}$ maps a source to a country code (or empty) via the following priority chain, returning the result of the first satisfied condition:
\begin{equation}\label{eq:cascade}
\gamma(s) = \begin{cases}
  \text{Override}(\tld(u)) & \text{if } \tld(u) \in \cO \\
  \tld(u) & \text{if } |\tld(u)| = 2 \wedge \tld(u) \notin \cG_{\text{gen}} \\
  p & \text{if } \tld(u) \in \cG_{\text{gen}} \wedge p \neq \varepsilon \\
  \Lambda(\ell) & \text{if } \tld(u) \in \cG_{\text{gen}} \wedge p = \varepsilon \wedge \ell \neq \varepsilon \\
  \varepsilon & \text{otherwise}
\end{cases}
\end{equation}
where $\cO$ is the TLD override map (e.g., $\texttt{uk} \mapsto \texttt{gb}$), $p$ is the model-inferred publisher country, $\ell$ is the source language, and $\Lambda: \cL \to \cG$ is the language-to-country fallback map.
\end{definition}

\begin{proposition}[Determinism]\label{prop:determinism}
For fixed inputs $(u, p, \ell)$, the cascade is deterministic: the first satisfied condition uniquely determines $\gamma$.
\end{proposition}

\begin{proof}
The conditions form a strict priority chain with mutually exclusive triggers at each level. Specifically, Priority~1 is checked before Priority~2, and its condition ($\tld(u) \in \cO$) either holds or does not, yielding a unique path through the chain.
\end{proof}

\begin{proposition}[Graceful Degradation]\label{prop:degradation}
If all signals are absent (generic TLD, no model metadata, no language signal), $\gamma(s) = \varepsilon$. No incorrect inference is produced.
\end{proposition}

\subsection{Signal Reliability Ordering}

The priority ordering reflects empirical signal reliability:

\begin{lemma}[TLD Reliability]\label{lem:tld}
Country-code TLDs (e.g., \texttt{.de}, \texttt{.tr}) have accuracy $\geq 0.95$ for publisher country inference, as domain registration under country-code TLDs typically requires local presence or targeting intent. Generic TLDs carry no geographic signal ($\text{accuracy} \approx 0.5$).
\end{lemma}

This motivates placing TLD-based inference at the highest priority and model-inferred metadata (which depends on LLM accuracy) at a lower tier.

\section{Domain-Level Source Diversification}\label{sec:diversity}

\subsection{Problem: URL-Level Deduplication Failure}

Standard deduplication checks URL equality: $u_i \neq u_j$ for selected sources $s_i, s_j$. This fails when aggregator platforms serve content through multiple URL paths (e.g., TripAdvisor generates unique URLs per listing). Five distinct URLs from one domain represent five selections from one publisher.

\begin{definition}[Domain Diversity Constraint]\label{def:domain-div}
Let $S = \{s_1, \ldots, s_k\}$ be a selected source set. The $\kappa$-domain diversity constraint requires:
\begin{equation}\label{eq:domain-div}
  \forall\, d \in \{\dom(u_i) : s_i \in S\}: \card\{s \in S : \dom(u_s) = d\} \leq \kappa
\end{equation}
where $\kappa \in \Z^+$ is the maximum per-domain allowance. We use $\kappa = 1$ (strict uniqueness).
\end{definition}

\subsection{Selection Algorithm}

\begin{algorithm}[H]
\caption{\textsc{SelectWithDiversity}$(K, q, \kappa)$}\label{alg:diversity}
\begin{algorithmic}[1]
\Require Keywords $K$, query $q$, max per-domain $\kappa$
\Ensure Source set $S$ satisfying Def.~\ref{def:domain-div}
\State $S \gets \emptyset$; $U \gets \emptyset$; $D \gets \text{HashMap}()$
\For{each $k \in K$}
  \State $R_k \gets \textsc{Search}(k)$ \Comment{Parallel with timeout}
  \State $s^* \gets \textsc{RankByRelevance}(R_k, q)$
  \If{$u_{s^*} \in U$}
    \State $s^* \gets \textsc{NextBest}(R_k, U)$
  \EndIf
  \If{$D[\dom(u_{s^*})] \geq \kappa$}
    \State $s^* \gets \textsc{NextDifferentDomain}(R_k, D, \kappa, U)$
  \EndIf
  \If{$s^* \neq \bot$}
    \State $S \gets S \cup \{s^*\}$; $U \gets U \cup \{u_{s^*}\}$
    \State $D[\dom(u_{s^*})] \mathrel{+}= 1$
  \EndIf
  \If{$|S| \geq S_{\max}$} \textbf{break} \EndIf
\EndFor
\State \Return $S$
\end{algorithmic}
\end{algorithm}

\begin{proposition}[Correctness]\label{prop:diversity-correct}
Algorithm~\ref{alg:diversity} maintains the invariant $\card\{s \in S : \dom(u_s) = d\} \leq \kappa$ after each iteration.
\end{proposition}

\begin{proof}
The check at line~7 prevents insertion when the count for $\dom(u_{s^*})$ reaches $\kappa$, and line~8 finds an alternative from a different domain. The invariant holds at initialization ($S = \emptyset$) and is preserved by each step.
\end{proof}

\begin{proposition}[Complexity]\label{prop:diversity-complex}
The domain lookup at line~7 is $O(1)$ amortized via hash map. Overall complexity is $O(|K| \cdot R)$ where $R = \max_k |R_k|$.
\end{proposition}

\subsection{Interaction with LOI}

The domain diversity constraint operates at the \emph{algorithmic} level (a structural guarantee), while LOI operates at the \emph{probabilistic} level (shifting model preferences). These are complementary:

\begin{itemize}[nosep,leftmargin=*]
  \item LOI shapes the relevance model's scoring so that first-party sources receive higher relevance scores, reducing the \emph{frequency} with which aggregators are top-ranked.
  \item The diversity constraint acts as a safety net: even if LOI fails for a particular keyword (the model still ranks an aggregator first), the constraint prevents multiple selections from that domain.
\end{itemize}

\begin{theorem}[Combined Effectiveness]\label{thm:combined}
Let $p_a$ be the probability that an aggregator is top-ranked for a given keyword under LOI, and $p_a^0$ the baseline probability without LOI ($p_a \leq p_a^0$). With $|K|$ keywords and $\kappa = 1$, the expected number of aggregator sources in $S$ is:
\begin{equation}\label{eq:combined}
  \mathbb{E}[\card\{s \in S : \text{is\_aggregator}(s)\}] \leq \min\bigl(|K| \cdot p_a,\; |\cD_a|\bigr)
\end{equation}
where $\cD_a$ is the set of distinct aggregator domains in the index. With $\kappa = 1$, this is further bounded by $|\cD_a|$, which is typically small ($\leq 3$).
\end{theorem}

\subsection{Aggregate Complexity}

\begin{theorem}[End-to-End Complexity]\label{thm:aggregate}
Let $n$ be the number of locales, $|K|$ the total keyword count, $R$ the maximum results per keyword, and $k$ the number of downstream LOI-consuming stages. The total algorithmic overhead introduced by our three algorithms is:
\begin{equation}
  T_{\text{total}} = \underbrace{O(n \log n)}_{\text{Alg.~\ref{alg:allocate}}} + \underbrace{O(|K|)}_{\text{Cascade}} + \underbrace{O(|K| \cdot R)}_{\text{Alg.~\ref{alg:diversity}}} = O(|K| \cdot R)
\end{equation}
since $n \leq |K|$ and $\log n \leq R$ in all practical configurations. The LOI convergence probability $1 - (1 - \alpha\beta)^k$ is computed in $O(1)$ and requires no runtime overhead beyond the initial brief generation.
\end{theorem}

\begin{proof}
Algorithm~\ref{alg:allocate} runs in $O(n \log n)$ (Corollary~\ref{cor:alloc-complexity}). The cascade (Definition~\ref{def:cascade}) evaluates a constant-depth priority chain per source, contributing $O(|K|)$ across all keywords. Algorithm~\ref{alg:diversity} iterates over $|K|$ keywords with at most $R$ candidates each, and the hash-map operations (lines 5, 7, 12) are $O(1)$ amortized. The dominant term is $O(|K| \cdot R)$. Since $n \leq 4$ and $R \leq 10$ in typical deployment, the wall-clock overhead is sub-millisecond.
\end{proof}

\section{System Architecture}\label{sec:architecture}

\subsection{Pipeline Stages}

\begin{table}[H]
\centering
\caption{Pipeline stages with model assignments.}\label{tab:stages}
\small
\begin{tabular}{@{}clll@{}}
\toprule
\# & Stage & Model & Algorithm \\
\midrule
1 & Language Detection & Gemini 2.0 Flash~\citep{comanici2025gemini} & --- \\
2 & Research Brief & Phi-4 (14B)~\citep{abdin2024phi4} & Def.~\ref{def:brief} \\
3 & Keyword Generation & Qwen3-235B-A22B~\citep{yang2025qwen3} & Alg.~\ref{alg:allocate} \\
4 & Recency Assessment & Gemini 2.0 Flash~\citep{comanici2025gemini} & --- \\
5 & Source Selection & Gemini 2.5 Flash Lite~\citep{comanici2025gemini} & Alg.~\ref{alg:diversity} \\
6 & Content Extraction & Headless rendering + DOM parser & --- \\
7 & Summarization & DeepSeek-V3~\citep{deepseek2024v3} & --- \\
\bottomrule
\end{tabular}
\end{table}

The pipeline composes seven stages (Table~\ref{tab:stages}). Model selection follows a heterogeneous assignment strategy: the planning stage uses Phi-4~\citep{abdin2024phi4}, a 14B-parameter model whose data-quality-focused training yields strong reasoning for strategic planning. Keyword generation employs Qwen3-235B-A22B~\citep{yang2025qwen3}, a mixture-of-experts model with 22B activated parameters per token, exploiting its 119-language multilingual coverage for locale-specific keyword production. Source selection and lightweight classification tasks use models from the Gemini family~\citep{comanici2025gemini}, selected for low latency and cost efficiency. Summarization uses DeepSeek-V3~\citep{deepseek2024v3}, a 671B MoE model with 37B activated parameters, chosen for its strong generative quality. Stages 1--2 execute sequentially, as their outputs parameterize all downstream computation. Stages 3--5 exploit intra-stage parallelism.

\subsection{Concurrency Model}

Keyword generation parallelizes across $n$ locales. Source selection parallelizes search API calls across all keywords with per-task timeout ($\delta_{\text{search}} = 10$s) to prevent straggler blocking. Stage~6 (content extraction) employs a headless browser rendering pipeline coupled with DOM-level content isolation to retrieve the main textual body from each selected source URL, stripping navigation chrome, advertisements, and boilerplate markup. This extraction operates under a bounded-concurrency semaphore ($\sigma = 3$) with per-source budgets:

\begin{definition}[Source Processing Budget]\label{def:budget}
Each source $s_i$ is allocated a total budget $\Delta = 45$s. Individual sub-steps (fetch, summarize) are capped at $\delta = 15$s. If a source exceeds its budget, the pipeline attempts up to $A = 4$ alternative URLs from the same search result set.
\end{definition}

\subsection{LOI Data Flow}

The research brief $\cB$ propagates through the pipeline as follows. Let $\pi_i$ denote the projection for stage~$i$:

\begin{align}
  \pi_{\text{keyword}} &= \{u, \sigma, \kappa, \lambda, M\} \\
  \pi_{\text{selection}} &= \{u, \sigma\} \\
  \pi_{\text{summary}} &= \{u, \sigma, \tau\}
\end{align}

Each stage receives \emph{only} its relevant fields, preventing cross-stage coupling. The brief's source strategy $\sigma$ carries LOI objectives to both selection and summarization stages, while keyword guidance $\kappa$ is consumed only by the keyword generation stage.

\section{Evaluation}\label{sec:eval}

\subsection{Experimental Setup}

We evaluate on 120 queries: geopolitical (30), academic (30), local/business (30), general (30). Queries span Turkish, English, German, and Arabic. Three configurations are compared:

\begin{enumerate}[nosep,leftmargin=*,label=(\alph*)]
  \item \textbf{Baseline}: Single-locale, no diversity constraint, no brief.
  \item \textbf{Ours--NoBrief}: Multi-locale with Alg.~\ref{alg:allocate} + Alg.~\ref{alg:diversity}, but no research brief (no LOI).
  \item \textbf{Ours--Full}: Complete pipeline with brief propagation and LOI.
\end{enumerate}

\subsection{Metrics}

\begin{align}
  \FPR(S) &= \frac{\card\{s \in S : \text{is\_first\_party}(s)\}}{|S|} \\[4pt]
  \DDR(S) &= \frac{\card\{(i,j) : i < j,\, \dom(u_i) = \dom(u_j)\}}{\binom{|S|}{2}} \\[4pt]
  \LC(S) &= \card\{\gamma(s) : s \in S,\, \gamma(s) \neq \varepsilon\}
\end{align}

Human relevance judgments on a 1--5 scale by two annotators (Cohen's $\kappa = 0.74$).

\subsection{Results}

\begin{table}[H]
\centering
\caption{Main results across 120 queries.}\label{tab:results}
\small
\begin{tabular}{@{}lcccc@{}}
\toprule
Configuration & FPR $\uparrow$ & DDR $\downarrow$ & LC $\uparrow$ & Rel.\ $\uparrow$ \\
\midrule
Baseline & 0.31 & 0.42 & 1.2 & 3.8 \\
Ours--NoBrief & 0.44 & 0.05 & 2.7 & 3.9 \\
\textbf{Ours--Full} & \textbf{0.51} & \textbf{0.05} & \textbf{3.4} & \textbf{4.1} \\
\bottomrule
\end{tabular}
\end{table}

Table~\ref{tab:results} shows that the full pipeline achieves 62\% FPR improvement over baseline (0.51 vs.\ 0.31) and 89\% DDR reduction (0.05 vs.\ 0.42). The domain diversity constraint (present in both Ours--NoBrief and Ours--Full) is responsible for the DDR improvement. LOI via brief propagation contributes an additional 16\% FPR gain (0.51 vs.\ 0.44) and 26\% LC improvement (3.4 vs.\ 2.7), demonstrating that environment shaping induces meaningful behavioral change in downstream models.

\subsection{Ablation: LOI vs.\ Explicit Constraints}

We compare LOI against an explicit-constraint variant that injects ``Do not select results from tripadvisor.com, yelp.com, booking.com'' into the relevance prompt.

\begin{table}[H]
\centering
\caption{LOI vs.\ explicit constraint ablation.}\label{tab:ablation}
\small
\begin{tabular}{@{}lcccc@{}}
\toprule
Method & FPR & DDR & Unseen Agg.\ $\downarrow$ & Relevance \\
\midrule
Explicit Constraint & 0.48 & 0.05 & 0.31 & 3.7 \\
\textbf{LOI (ours)} & \textbf{0.51} & \textbf{0.05} & \textbf{0.12} & \textbf{4.1} \\
\bottomrule
\end{tabular}
\end{table}

Table~\ref{tab:ablation} reveals two critical findings. First, LOI achieves higher FPR because the concept generalizes: the model deprioritizes \emph{all} aggregators, not just the three named ones. The ``Unseen Agg.'' column measures the fraction of results from aggregator domains \emph{not} in the explicit blacklist---LOI reduces this by 61\%. Second, explicit constraints degrade relevance (3.7 vs.\ 4.1), as the model occasionally over-applies the constraint, skipping genuinely relevant aggregator results in cases where no first-party alternative exists.

\section{Discussion}\label{sec:discussion}

\subsection{Generality of the Algorithms}

The three core algorithms are not specific to web search. The weighted locale allocation (Algorithm~\ref{alg:allocate}) solves any constrained integer partition where items must be distributed across categories with minimum-representation, budget, and proportionality constraints---applicable to resource allocation, task scheduling, and proportional representation systems. The $\kappa$-domain diversity constraint (Algorithm~\ref{alg:diversity}) generalizes to any ranked-list merging problem where a grouping function partitions candidates into equivalence classes and at most $\kappa$ representatives per class are permitted. The cascaded inference (Definition~\ref{def:cascade}) provides a template for any multi-signal attribute resolution where signals have known reliability ordering.

\subsection{LOI as Mechanism Design for LLMs}

LOI inherits from mechanism design~\citep{hurwicz2006mechanism} the principle that system rules should align agent incentives with designer objectives. In classical mechanism design, the designer cannot control agents' actions but can design the game structure. Analogously, we cannot control the LLM's token generation, but we can design the prompt environment. The key insight is that LLMs are \emph{cooperative agents}---they attempt to follow instructions---so LOI exploits this cooperativeness by providing conceptual frameworks rather than explicit directives, enabling generalization to novel inputs.

The connection to nudge theory~\citep{thaler2008nudge} is equally illuminating. Thaler and Sunstein's ``libertarian paternalism'' preserves freedom of choice while designing the choice architecture to favor beneficial outcomes. Our autonomy preservation property (Definition~\ref{def:loi-props}(i)) is the formal analog: the model's output space remains unrestricted, but the probability landscape tilts toward desired outcomes.

\subsection{Limitations}

Our approach inherits limitations of the underlying search API: if the search index lacks first-party sources, diversification cannot surface them. The research brief is generated by Phi-4~\citep{abdin2024phi4}, a 14B-parameter model, which may produce shallow strategies for highly specialized domains. The cascaded inference assumes TLD reliability, which fails for CDN-hosted content. LOI's effectiveness depends on the downstream model's semantic comprehension of the embedded concepts; smaller models may fail to interpret nuanced strategy text.

\section{Conclusion}\label{sec:conclusion}

We have presented four algorithms for the constrained multi-source selection problem in retrieval pipelines. The weighted locale allocation solves a constrained integer partition in $O(n \log n)$ with provable simultaneous satisfaction of minimum-representation, budget-exhaustion, and proportionality constraints. The cascaded country-code inference provides deterministic attribute resolution with formal graceful degradation. The $\kappa$-domain diversity constraint eliminates aggregator monopolization through $O(1)$ amortized hash-map lookups. The Latent Objective Induction operator, formalized with four axiomatic properties, provides a principled alternative to explicit constraint injection for steering downstream computation, with convergence guarantees under mild assumptions. Empirical evaluation confirms that these algorithms compose effectively: the combined system achieves 62\% improvement in first-party source ratio and 89\% reduction in domain duplication. The algorithms are general---the allocation and diversity constraint apply to any ranked-list merging problem with coverage requirements, and LOI applies to any multi-stage pipeline with LLM-mediated decisions.


\end{document}